\documentclass[a4paper,11pt]{article}
\usepackage{fullpage}
\usepackage{amsmath}
\usepackage{amsthm}
\usepackage{amsfonts}
\usepackage{mathrsfs}
\usepackage{amssymb}
\usepackage{algorithm}
\usepackage{algorithmic}
\usepackage{cite}

\usepackage{tikz}
\usepackage{pgfplots}
\pgfplotsset{compat=1.5}

\title{Quantum Algorithm for Triangle Finding in Sparse Graphs}
\author{
Fran{\c c}ois Le Gall  \qquad\qquad
Shogo Nakajima\vspace{2mm}\\
Department of Computer Science\\
Graduate School of Information Science and Technology\\
The University of Tokyo, Japan
}

\date{}
\newtheorem{theorem}{Theorem}

\newtheorem{proposition}{Proposition}[section]
\newtheorem{definition}{Definition}[section]
\newtheorem{lemma}{Lemma}[section]
\newcommand{\mydelta}[1]{\Delta_G(X,#1)}
\newcommand{\E}{\mathcal{E}}

\newcommand{\Ss}{\mathcal{S}}
\newcommand{\polylog}{\mathrm{polylog}}
\newcommand{\poly}{\mathrm{poly}}
\newcommand{\ket}[1]{\vert #1 \rangle}
\newcommand{\ceil}[1]{\left\lceil #1 \right\rceil}

\begin{document}

\maketitle

\begin{abstract}
This paper presents a quantum algorithm for triangle finding over sparse graphs that improves over the previous best quantum algorithm for this task by Buhrman et al. [SIAM Journal on Computing, 2005]. Our algorithm is based on the recent $\tilde O(n^{5/4})$-query algorithm given by Le Gall [FOCS 2014] for triangle finding over dense graphs (here~$n$ denotes the number of vertices in the graph). We show in particular that triangle finding can be solved with $O(n^{5/4-\epsilon})$ queries for some constant $\epsilon>0$ whenever the graph has at most $O(n^{2-c})$ edges for some constant $c>0$.  
\end{abstract}
\section{Introduction}
\paragraph{Background.}
Triangle finding asks to decide if a given undirected graph $G=(V,E)$ contains a cycle of length three, i.e., whether there exist three vertices $u_1,u_2,u_3\in V$ such that 
$\{u_1,u_2\}\in E$, $\{u_1,u_3\}\in E$ and $\{u_2,u_3\}\in E$. This problem has received recently a lot of attention, for the following reasons.

First, several new applications of triangle finding have been discovered recently. In particular, Vassilevska Williams and Williams have shown a surprising reduction from Boolean matrix multiplication to triangle finding \cite{Williams+FOCS10}, which indicates that efficient algorithms for triangle finding may be used to design efficient algorithms for matrix multiplication, and thus also for a vast class of problems related to matrix multiplication. Relations between variants of the standard triangle finding problem (such as triangle finding over weighted graphs) and well-studied algorithmic problems (such as 3SUM)  have also been shown in the past few years  (see for instance \cite{PatrascuSTOC10,Vassilevska+SICOMP13}).

Second, triangle finding is one of the most elementary graph theoretical problems whose complexity is unsettled. In the time complexity setting, the best classical algorithm uses a reduction to matrix multiplication \cite{Itai+SICOMP78} and solves triangle finding in time $O(n^{2.38})$, where $n$ denotes the number of vertices in $G$. In the time complexity setting again, Grover search \cite{GroverSTOC96} immediately gives, when applied to triangle finding as a search over the set of triples of vertices of the graph, a quantum algorithm with time complexity $\tilde O(n^{3/2})$, which is still the best known upper bound for the quantum time complexity of this problem.\footnote{In this paper the notation $\tilde O(\cdot)$ removes $\polylog n$ factors.} In the query complexity setting, where an oracle to the adjacency matrix of the graph is given and only the number of calls to this oracle is counted, a surge of activity has lead to quantum algorithms with better complexity. Magniez, Santha and Szegedy \cite{Magniez+SICOMP07} first presented a quantum algorithm that solves triangle finding with $\tilde O(n^{1.3})$ queries. This complexity was later improved to $O(n^{1.296\ldots})$ by Belovs \cite{BelovsSTOC12}, then to $O(n^{1.285\ldots})$ by Lee, Magniez and Santha \cite{Lee+SODA13} and Jeffery, Kothari and Magniez \cite{Jeffery+SODA13}, and further improved recently to $\tilde O(n^{5/4})$ by Le Gall \cite{LeGallFOCS14}. The main open problem now is to understand whether this $\tilde O(n^{5/4})$-query upper bound is tight or not. The best known lower bound on the quantum query complexity of triangle finding is the straightforward $\Omega(n)$ lower bound. 

Another reason why triangle finding has received much attention from the quantum computing community is that 
work on the quantum complexity of triangle finding has been central to the development of algorithmic techniques. Indeed, all the improvement mentioned in the previous paragraph have been obtained by introducing either new quantum techniques or new paradigms for the design of quantum algorithms: applications of quantum walks to graph-theoretic problems \cite{Magniez+SICOMP07}, introduction of the concept of learning graphs \cite{BelovsSTOC12} and improvements to this technique \cite{Lee+SODA13}, introduction of quantum walks with quantum data structures \cite{Jeffery+SODA13}, association of combinatorial arguments with quantum walks \cite{LeGallFOCS14}.

\paragraph{Triangle finding in sparse graphs.}
The problem we will consider in this paper is triangle finding over sparse graphs (the graphs considered are, as usual, undirected and unweighted). If we denote $m$ the number of edge of the graph (i.e., $m=|E|$), the goal is to design algorithms with complexity expressed as a function of $m$ and $n$. Ideally, we would like to show that if $m=n^{2-c}$ for any constant $c>0$ then triangle finding can be solved significantly faster than in the dense case (i.e., $m\approx n^2$). Besides its theoretical interest, this problem is of practical importance since in many applications the graphs considered are sparse.  

Classically, Alon, Yuster and Zwick \cite{Alon+97} constructed an algorithm exploiting the sparsity of the graph and working in time $O(m^{1.41})$, which gives better complexity than the $O(n^{2.38})$-time complexity mentioned above when $m\le n^{1.68}$. Understanding whether an improvement over the dense case is also possible for larger values $m$ is a longstanding open problem. Note in the classical query complexity setting it is easy to show that the complexity of triangle finding is $\Theta(n^2)$, independently of the value of~$m$.

In the quantum setting, using amplitude amplification, Buhrman et al.~\cite{Buhrman+SICOMP05} showed how to construct a quantum algorithm for triangle finding with time and query complexity $O(n+\sqrt{nm})$. This upper bound is tight when $m\le n$ since the $\Omega(n)$-query lower bound for the quantum query complexity of triangle finding already mentioned also holds when $m$ is a constant. 
Childs and Kothari \cite{Childs+SICOMP12} more recently developed an algorithm, based on quantum walks, that detects the existence of subgraphs in a given graph. Their algorithm works for any constant-size subgraph. For detecting the existence of a triangle, however, the upper bound they obtain is $\tilde O(n^{2/3}\sqrt{m})$ queries for $m\ge n$, which is worse that the bound obtained in \cite{Buhrman+SICOMP05}. Buhrman et al.'s result in particular gives an improvement over the $\tilde O(n^{5/4})$-query quantum algorithm algorithm whenever $m\le n^{3/2}$. A natural question is whether a similar improvement can be obtained for larger values of $m$. For instance, can we obtain query complexity $\tilde O(n^{5/4-\epsilon})$ for some constant $\epsilon>0$ when $m\approx n^{1.99}$? A positive answer would show that even a little amount of sparsity can be exploited in the quantum query setting, which is not known to be true in the classical setting as mentioned in the previous paragraph.

\paragraph{Our results.}
In this paper we answer positively to the above question. Our main result is as follows.
\begin{theorem}\label{th:main}
There exists a quantum algorithm that solves, with high probability, the triangle finding problem over graphs of $n$ vertices and $m$ edges with query complexity
\[
\begin{cases}
O(n+\sqrt{nm}) & \textrm{ if }\:\: 0 \leq m \leq n^{7/6}, \\
\tilde O(nm^{1/14}) & \textrm{ if }\:\: n^{7/6} \leq m \leq n^{7/5}, \\
\tilde O(n^{1/6}m^{2/3}) & \textrm{ if }\:\:n^{7/5} \leq m  \leq n^{3/2}, \\
\tilde O(n^{23/30}m^{4/15}) & \textrm{ if }\:\:n^{3/2} \leq m  \leq n^{13/8}, \\
\tilde O(n^{59/60}m^{2/15}) & \textrm{ if }\:\: n^{13/8} \leq m \leq n^2.
\end {cases}
\]
\end{theorem}
The complexity bounds of Theorem \ref{th:main} are depicted in Figure \ref{fig1}. 
For the dense case (i.e., $m\approx n^2$) we recover the same complexity $\tilde O(n^{5/4})$ as in \cite{LeGallFOCS14} --- in this case it turns out that our algorithm applies exactly the same procedure as in \cite{LeGallFOCS14}.
Whenever $m=n^{2-c}$ for some constant $c>0$ (in particular, for $m\approx n^{1.99}$), we indeed obtain query complexity $\tilde O(n^{5/4-\epsilon})$ for some constant $\epsilon>0$ depending on~$c$. The query complexity of our algorithm is better than the query complexity of Buhrman et al.'s algorithm \cite{Buhrman+SICOMP05} whenever $m\gtrsim n^{7/6}$. When $m\lesssim n^{7/6}$ we obtain the same complexity $O(n+\sqrt{nm})$ as in \cite{Buhrman+SICOMP05} --- in this case it turns out that our algorithm applies exactly the same procedure as in \cite{Buhrman+SICOMP05}.
 
\begin{figure}[ht]
\centering
\begin{tikzpicture}
\begin{axis}[
legend cell align=left,
width=9.5cm, 
height=6.5cm,
xmin=0.9, xmax=2.1,
thick,
 scale only axis,
xmajorgrids,
ymajorgrids,
 xtick={1,1.16666,1.4,1.5,1.625,2},
 xticklabels={$n$,$n^{7/6}$, $n^{7/5}$,$n^{3/2}$,$n^{13/8}$,$n^2$},
  ytick={1,1.1,1.16666,1.2,1.25,1.5,1.66666},
 yticklabels={$n$, $n^{11/10}$,$n^{7/6}$,$n^{6/5}$,$n^{5/4}$,$n^{3/2}$,$n^{5/3}$},
xlabel={Number of edges $m$},
ylabel={Query complexity},
legend pos=south east]

\addplot [solid, every mark/.append style={solid, fill=gray}] coordinates {
(0.5,1)
(1,1)
(1.166666,1.08333)
(1.4,1.1)
(1.5,1.166666)
(1.625,1.2)
(2,1.25)
};

\addplot [dotted, every mark/.append style={solid, fill=gray}] coordinates {
(0.5,1)
(1,1)
(2,1.5)
};

\addplot [only marks, every mark/.append style={solid, fill=gray}, mark=otimes*] coordinates {
(1.166666,1.08333)
(1.4,1.1)
(1.5,1.16666)
(1.625,1.2)
(2,1.25)
(1,1)
(2,1.5)
};

\legend{Our algorithm (Theorem \ref{th:main}), Buhrman et al. \cite{Buhrman+SICOMP05}}
\end{axis}
\end{tikzpicture}
\caption{\label{fig1}Quantum query complexity of triangle finding on a graph with $n$ vertices and $m$ edges.}

\end{figure}
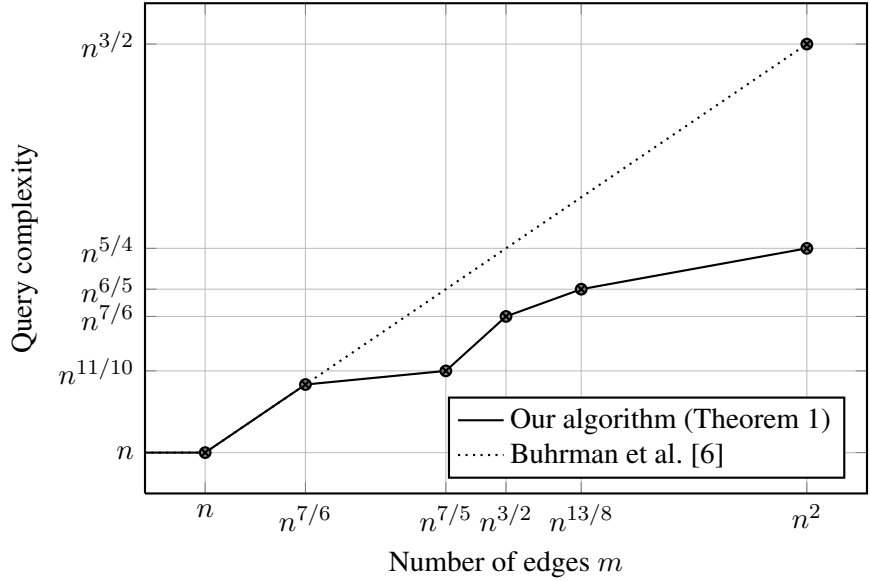

\setlength\textfloatsep{8pt}

\paragraph{Overview of our techniques.}
The main idea is to adapt the $\tilde O(n^{5/4})$-query quantum algorithm for triangle finding \cite{LeGallFOCS14} to handle sparse graphs. This algorithm works in two steps: a first step based on Grover search that detects the existence of a triangle in a well-chosen small part of the graph $G$, and a second step based on recursive quantum walks that detects the existence of a triangle in the remaining (large) part of the graph. A first simple observation is that the first step can be implemented faster in the case of sparse graphs by applying the quantum algorithm by Buhrman et al.~\cite{Buhrman+SICOMP05} on the small part of~$G$ instead of using Grover search. This observation alone however does not give any interesting speed-up unless sparsity is exploited in the second step as well. The hard part is actually to adapt the recursive quantum walk approach to the case of sparse graphs, and we outline below our main ideas to achieve this goal. 

The main issue is that the implementation of the recursive quantum walks described \cite{LeGallFOCS14} is really tailored for their application on dense graphs; when trying to use the same implementation for sparse graphs prohibitive intermediate costs (of order $\tilde O(n^{5/4})$, which is fine for the dense case, but not for the sparse case) appear. To overcome this difficulty, we need to modify partially the original approach in several ways, such as modifying how the inner quantum walk checks if it has found a solution and adjusting how the sets of marked states of the walk are defined, to fully exploit the sparsity of the graph. 

Several more technical issues have also to be dealt with. The complexity of a quantum walk basically depends on the complexity of three operations performed by the walk: the set up cost (creating the data structure corresponding to the initial state of the walk), the update cost (updating the database after updating the current state of the walk) and the checking cost (checking if the current state of the walk is marked or not). The sparsity of the graph $G$ can be immediately exploited to reduce the cost of these three operations if the graph is ``perfectly balanced'', i.e., if each vertex of the graph has degree $\Theta (m/n)$. However, while the average degree will indeed be $\Theta (m/n)$, in general the graph can have many vertices with degree exceeding this estimate (in particular this can happen for vertices of the triangle we are looking for). This is a significant complication since to analyze a quantum walk one need an upper bound on the worst case (i.e., for the worse state of the walk) complexity of the three operations. Indeed, there is no general technique to analyze quantum walks when only an upper bound on the average update cost or checking cost is available. To overcome this difficulty, our approach is to partition the vertices of $V$ into two sets: the set of vertices with degree larger than $n^d$ (which we call below high-degree vertices) and the set of vertices with degree smaller than $n^d$ (low-degree vertices), where $d$ is a parameter. Obtaining the classification can be done by combining quantum search and quantum counting, but is costly when~$d$ is small, which means that we need to be careful when choosing $d$. Once this classification has been obtained, we only need to search separately for four types of triangles: triangles with three low-degree vertices, triangles with two low-degree vertices and one high-degree vertex, triangles with one low-degree vertex and two high-degree vertices, and triangles with three high-degree triangles. Since we know that each low-degree vertex has degree at most $n^d$, we can derive a worst-case upper bound for the corresponding update costs and checking costs. For high-degree vertices we do not have any upper bound on there degree, but we know that the number of high-degree vertices is at most $m/n^d$, which will be significantly smaller than $n$ if $d$ is well chosen and lead to some improvement for the corresponding complexity of the walks, since the graph has at most $m$ edges. Combined with the ideas described in the previous paragraph, this strategy enables us to obtain the upper bounds given in Theorem \ref{th:main}.

\section{Preliminaries}
\subsection{Query complexity for graph-theoretic problems}\label{sec:query}
In this paper we adopt the standard model of quantum query complexity for graph-theoretic problems. The presentation given below will follow the description of this notions given in \cite{LeGallFOCS14}.

For any finite set $T$ and any $r\in\{1,\ldots,|T|\}$ we denote $\Ss(T,r)$ the set of all subsets of~$r$ elements 
of $T$. We use the notation $\E(T)$ to represent $\Ss(T,2)$,
i.e., the set of unordered pairs of elements in $T$.

Let $G=(V,E)$ be an undirected and unweighted graph, where
$V$ represents the set of vertices and $E\subseteq\E(V)$ 
represents the set of edges. We write $n=|V|$.
 In the query complexity setting, 
we assume that $V$ is known, and that~$E$ can be accessed through 
a quantum unitary operation $\mathcal{O}_G$ defined as follows.
For any pair $\{u,v\}\in \E(V)$, any bit $b\in\{0,1\}$, and any 
binary string $z\in\{0,1\}^\ast$, the operation $\mathcal{O}_G$ 
maps the basis state 
$\ket{\{u,v\}}\ket{b}\ket{z}$ to the state
\[
\mathcal{O}_G\ket{\{u,v\}}\ket{b}\ket{z}=\left\{
\begin{array}{ll}
\ket{\{u,v\}}\ket{b\oplus 1}\ket{z}&\textrm{ if } \{u,v\}\in E,\\
\ket{\{u,v\}}\ket{b}\ket{z}&\textrm{ if } \{u,v\}\notin E,\\
\end{array}
\right.
\]
where $\oplus$ denotes the bit parity (i.e., the logical XOR).
We say that a quantum algorithm computing some property of $G$
uses $k$ queries if the operation $\mathcal{O}_G$, given as an oracle, is called $k$ times by the algorithm. We also assume that we know the number of edges of the input graph (i.e., we know $m=|E|$). All the results in this paper can be easily generalized to the case where $m$ is unknown.

\paragraph{Quantum enumeration.}
Let $f_G\colon \{1,\ldots,N\}\to \{0,1\}$ be a Boolean function depending on the input graph $G$, and let us write $M=f^{-1}(1)$.
Assume that for any $x\in\{1,\ldots,N\}$ the value $f_G(x)$ can be computed using at most $t$ queries to $\mathcal{O}_G$.
Grover search enables us to find an element $x$ such that $f_G(x)=1$, if such an element exists, using $\tilde O(\sqrt{N/M}\times t)$ queries to $\mathcal{O}_G$. A folklore observation is that we can then repeat this procedure to find all the elements $x\in \{1,\ldots,N\}$ such that $f_G(x)=1$ with 
$
\tilde O\left(\left(\sqrt{\frac{N}{M}}
+\sqrt{\frac{N}{M-1}}+\cdots+\sqrt{\frac{N}{1}}
\right)\times t\right)
=\tilde O\left(\sqrt{N\times M}\times t\right)
$
queries.
We call this procedure \emph{quantum enumeration}.

\paragraph{Quantum walk over Johnson graphs.}
Let $T$ be a finite set and $r$ be a positive integer such that 
$r\le |T|$. Let $f_G\colon \Ss(T,r)\to \{0,1\}$ be a Boolean function depending on the input graph $G$. We say that a set $A\in\Ss(T,r)$ is marked if $f_G(A)=1$. Let us consider the following problem. The goal is to find a marked set, if such a set exists, or otherwise report that there is no marked set. We are interested in the number of calls to $\mathcal{O}_G$ to solve this problem. The quantum walk search approach developed by Ambainis \cite{AmbainisSICOMP07} solves this problem using a quantum walk over a Johnson graph.

The Johnson graph $J(T,r)$ is the undirected graph with vertex set $\Ss(T,r)$ where two vertices $R_1,R_2\in\Ss(T,r)$ are connected if and only if $|R_1\cap R_2|=r-1$. In a quantum walk over a Johnson graph $J(T,r)$, the state of the walk corresponds to a node of the Johnson (i.e., to an element $A\in\Ss(T,r)$). A data structure $D(A)$, which in general depends on $G$, is associated to each state $A$. There are three costs to consider: the set up cost $\mathsf{S}$ representing the number of queries to $\mathcal{O}_G$ needed to construct the data structure of the initial state of the walk, the update cost $\mathsf{U}$ representing the number of queries to $\mathcal{O}_G$ needed to update the data structure when one step of the quantum walk is performed (i.e., updating $D(A)$ to $D(A')$ for some $A'\in\Ss(T,r)$ such that $|A\cap A'|=r-1$), and the checking cost $\mathsf{C}$ representing the number of queries to $\mathcal{O}_G$ needed to check if the current state $A$ is marked (i.e., checking whether $f_G(A)=1$). Let $\varepsilon>0$ be such that, for all input graphs $G$ for which at least one marked set exists, the fraction of marked states is at least $\varepsilon$. Ambainis \cite{AmbainisSICOMP07} (see also \cite{Magniez+SICOMP11}) has shown that the quantum walk search approach outlined above finds with high probability a marked set if such set exists (or otherwise report that there is no marked set)
and has query complexity 
$
\tilde O\left(\mathsf{S}+\frac{1}{\sqrt{\varepsilon}}\left(\sqrt{r}\times \mathsf{U}+\mathsf{C}\right)\right).
$

\subsection{Quantum algorithm for dense triangle finding}\label{sec:dense}
In this subsection we outline the $\tilde O(n^{5/4})$-query quantum algorithm for triangle finding over a dense graph by Le Gall \cite{LeGallFOCS14}. We actually present a version of this algorithm that solves the following slightly more general version of triangle finding, since this will be more convenient when describing our algorithms for sparse graphs in the next section: given two (non necessarily disjoint) sets $V_1,V_2\subseteq V$, find a triangle $\{v_1,v_2,v_3\}$ of $G$ such that $v_1\in V_1$ and $v_2,v_3\in V_2$, if such a triangle exists. Note that the original triangle finding problem is the special case $V_1=V_2=V$.

\paragraph{Definitions and lemmas.}
Let $V_1$ be any subset of $V$.  
For any sets $X\subseteq V_1$ and $Y\subseteq V$, 
we define the set $\mydelta{Y}\subseteq \E(Y)$
as follows:
\[
\mydelta{Y}=\E(Y)\setminus \bigcup_{u\in X}\E(N_G(u)),
\]
where $N_G(u)$ denotes the set of neighbors of $u$.
For any vertex $w\in V$, we define the set $\mydelta{Y,w}\subseteq \mydelta{Y}$
as follows:
\begin{align*}
\mydelta{Y,w}&=
\Big\{\{u,v\}\in \mydelta{Y}\:|\: \{u,w\}\in E \textrm { and } \{v,w\}\in E\Big\}.
\end{align*}

An important concept used in \cite{LeGallFOCS14} is the notion of \emph{$k$-good sets}.

\begin{definition}\label{def:good}
Let $k$ be any constant such that $0\le k\le 1$, and $V_1$ be any subset of $V$. 
A set $X\subseteq V_1$ is $k$-good for $(G,V_1)$ if 
the inequality
$
\sum_{w\in V_1}\left|\mydelta{Y,w}\right|\le |Y|^2 |V_1|^{1-k}
$
holds for all $Y\subseteq V$.
\end{definition}
Note that \cite{LeGallFOCS14} considered only Definition \ref{def:good} for the case $V_1=V$. In our paper we will need the slightly generalized version described here. The point is that $k$-good sets can be constructed very easily.
\begin{lemma}[\cite{LeGallFOCS14}]\label{lemma:sparse}
Let $k$ be any constant such that $0\le k\le 1$. Suppose that $X$ is a set obtained 
by taking uniformly at random, with replacement,  $\ceil{3|V_1|^{k}\log n}$ elements from $V_1$.
Then $X$ is $k$-good for $(G,V_1)$ with probability at least $1-1/n$.
\end{lemma}
Lemma~\ref{lemma:sparse} was proved in \cite{LeGallFOCS14} only for the case $V_1=V$, but the generalization is straightforward.

\paragraph{Quantum algorithm for dense triangle finding.}
Let $a$, $b$ and $k$ be three constants such that $0<b< a<1$ and $0<k<1$. The values of these constants will be set later. The quantum algorithm in \cite{LeGallFOCS14} works as follows.

The algorithm first takes a set $X\subseteq V_1$ obtained 
by choosing uniformly at random $\ceil{3|V_1|^{k}\log n}$ elements from $V_1$,
and checks if there exists a triangle of~$G$ with a vertex in $X$ and two vertices in $V_2$.
This can be done
using Grover search with
\begin{equation}\label{eq0}
O\left(\sqrt{|X|\times |\E(V_2)|}\right)=\tilde O\left(|V_1|^{k/2}|V_2|\right)
\end{equation}
queries. If no triangle has been reported, 
we know that any triangle of $G$ with one vertex in $V_1$ and two vertices in $V_2$
must have an edge in $\mydelta{V_2}$.
 
Now, in order to find a triangle with an edge in $\mydelta{V_2}$,
if such a triangle exists, the idea is to search for a set $A\in\Ss(V_2,\ceil{|V_2|^a})$ such that
$\mydelta{A}$ contains an edge of a triangle. 
To find such a set $A$, the algorithm performs a quantum walk over the Johnson graph
$J(V_2,\ceil{|V_2|^a})$.
The states of this walk correspond to the elements in $\Ss(V_2,\ceil{|V_2|^{a}})$.
The state corresponding to a set $A\in \Ss(V_2,\ceil{|V_2|^{a}})$ is  
marked if $\mydelta{A}$ contains an edge of a triangle of $G$. 
In case the set of marked states is not empty, 
the fraction of marked states is 
\[
\varepsilon=\Omega\left(|V_2|^{2(a-1)}\right).
\]
The data structure of the walk stores the set $\mydelta{A}$.
Concretely, this is done by 
storing the couple $(v,N_G(v)\cap X)$ for each $v\in A$, since this information is enough to construct
$\mydelta{A}$ without using any additional query.
The setup cost is 
$\mathsf{S}= |A|\times |X|=\tilde O(|V_2|^{a}|V_1|^k)$ queries.
The update cost is $\mathsf{U}=2|X|=\tilde O(|V_1|^{k})$ queries. 
The query complexity of the quantum walk is 
\begin{equation}\label{eq1}
\tilde O\left(\mathsf{S}+\sqrt{1/\varepsilon}\left(|V_2|^{a/2}\times \mathsf{S}+\mathsf{C}\right)\right),
\end{equation}
where $\mathsf{C}$ is the cost of checking if a state is marked.

The checking procedure is done as follows: check if there exists a vertex $w\in V_1$ such that $\mydelta{A}$ contains a pair $\{v_1,v_2\}$ for which $\{v_1,v_2,w\}$ is a triangle of~$G$.
For any $w\in V_1$, let $Q(w)$ denote the query complexity of checking if there exists a pair
$\{v_1,v_2\}\in \mydelta{A}$ such that $\{v_1,v_2,w\}$ is a triangle of~$G$. Using Ambainis' variable cost search \cite{Ambainis10} this checking procedure can be implemented using
\[
\mathsf{C}=\sqrt{\sum_{w\in V_1} Q(w)^2}
\]
queries. It thus remains to give an upper bound on $Q(w)$.
Let us fix $w\in V_1$. First, a tight estimator of the size of $\mydelta{A,w}$ is computed: the algorithm computes an integer $\delta(X,A,w)$ such that 
$|\delta(X,A,w)-|\mydelta{A,w}||\le \frac{1}{10}\times |\mydelta{A,w}|$, which can be done in $\tilde O(|V_1|^k)$ queries using (classical) sampling.
The algorithm then performs a quantum walk over the Johnson graph $J(A,\ceil{|V_2|^{b}})$.
The states of this walk correspond to the elements in $\Ss(A,\ceil{|V_2|^{b}})$.
We now define the set of marked states of the walk.
The state corresponding to a set 
$B\in \Ss(A,\ceil{|V_2|^{b}})$ is marked if $B$ satisfies the following two conditions:
\begin{itemize}
\setlength{\leftskip}{0.3cm}
\item[(i)]
there exists a pair $\{v_1,v_2\}\in \mydelta{B,w}$ such that $\{v_1,v_2\}\in E$ (i.e., such that $\{v_1,v_2,w\}$ is a triangle of~$G$); 
\item[(ii)]
$|\mydelta{B,w}|\le 10\times |V_2|^{2(b-a)}\times \delta(X,A,w)$.
\end{itemize}
The fraction of marked states is 
\[
 \varepsilon'=\Omega\left(|V_2|^{2(b-a)}\right).
\]
The data structure of the walk will store $\mydelta{B,w}$. Concretely, this is done by 
storing the couple $(v,e_v)$ for each $v\in B$, where $e_v=1$ if $\{v,w\}\in E$
and $e_v=0$ if $\{v,w\}\notin E$.
The setup cost is $\mathsf{S}'=\ceil{|V_2|^b}$ queries since it is sufficient to check if $\{v,w\}$ is an edge for all $v\in B$.
The update cost is $\mathsf{U}'=2$ queries. The checking cost is
\[
\mathsf{C}'_w\!=\!O\!\left(\sqrt{|\mydelta{B,w}|}\right)\!=\! O\!\left(\frac{|V_2|^b}{|V_2|^a}\sqrt{\delta(X,A,w)}\right)\!=\! O\!\left(\frac{|V_2|^b}{|V_2|^a}\sqrt{|\Delta(X,A,w)|}\right)\!\!.
\]
We thus obtain the bound
\[
Q(w)=\tilde O\left(|V_1|^k+\mathsf{S}'+\sqrt{1/\varepsilon'}\left(|V_2|^{b/2}\times \mathsf{U}'+\mathsf{C}'_w\right)\right),
\]
and conclude that
\[
\mathsf{C}\!=\!\tilde O\!\left(\!\sqrt{|V_1|}\!\left(|V_1|^k\!+\!\mathsf{S}'\!+\!\frac{|V_2|^{b/2}\times \mathsf{U}'}{\sqrt{\varepsilon'}}\right)\!\!+\!\!\frac{|V_2|^{b-a}}{\sqrt{\varepsilon'}}\!\!\times\!\!\sqrt{\sum_{w\in V_1} |\Delta(X,A,w)|}\right)\!.
\]

The final key observation is that, since the set $X$ is $k$-good for $(G,V_1)$ with high probability, as guaranteed by Lemma \ref{lemma:sparse}, the term $\sum_{w\in V} |\Delta(X,A,w)|$ in the above expression can be replaced by $O(|V_2|^{2a}|V_1|^{1-k})$, which enables us to express $\mathsf{C}$ as a function of $a$, $b$ and $k$, and then the complexity of the second part of the algorithm (Expression (\ref{eq1})) as a function of $a$, $b$ and $k$. The complexity of the whole algorithm  (the maximum of Expression (\ref{eq0}) and Expression (\ref{eq1})) can thus be written as a function of $a$, $b$ and $k$ as well.

For the original triangle finding problem (i.e., for the case $V_1=V_2=V$), taking $a=\frac{3}{4}$ and $b=k=\frac{1}{2}$ gives query complexity $\tilde O(n^{5/4})$.

\section{Quantum Algorithm for Sparse Triangle Finding}
In this section we describe our quantum algorithm for triangle finding in sparse graphs and prove Theorem \ref{th:main}. 

Let $d$ be a real number such that $0\le d\le 1$. The value of this parameter will be set later. Define the following two subsets of $V$:
\begin{align*}
\mathcal{V}^{d}_h &= \{v \in V \mid \deg(v) \ge \frac{9}{10}\times n^{d}\}, \\
\mathcal{V}^{d}_l &= \{v \in V \mid \deg(v) \leq \frac{11}{10}\times n^{d}\}.
\end{align*}
A crucial observation is that $|\mathcal{V}^d_h|=O(m/{n^d})$, since the graph $G$ has $m$ edges.
The following proposition shows how to efficiently  classify all the vertices of $V$ into vertices in $\mathcal{V}^{d}_h$ and vertices in $\mathcal{V}^{d}_l$.
\begin{proposition}\label{prop:classification}
There exists a quantum algorithm using $Q_1 = \tilde{O}(n^{1-d}\sqrt{m})$ queries that partitions the set $V$ into two sets $V^{d}_h$ and $V^{d}_l$ such that, with high probability, $V^{d}_h\subseteq \mathcal{V}^{d}_h$ and $V^{d}_l\subseteq \mathcal{V}^{d}_l$.
\end{proposition}                                                       
\begin{proof}                                                           
Let $v$ be any vertex in $V$.
Using quantum counting \cite{Brassard+ICALP98} we can compute, using $\tilde O\left(\sqrt{\frac{n}{n^{d}}}\right)$ queries, a value $a(v)$ such that $|a(v)-\deg(v)|\le n^d/100$ with probability at least $1-1/\poly(n)$. We use $a(v)$ to classify $v$ as follows: we decide `` $v$ is in $\mathcal{V}^{d}_h$ '' if  $a(v)\ge n^d$, and decide `` $v$ is in $\mathcal{V}^{d}_l$ '' if  $a(v)< n^d$. This decision is correct with probability at least $1-1/\poly(n)$.

We can thus apply quantum enumeration as described in Section \ref{sec:query} to obtain a set $V_h^d\subseteq V$ of vertices such that, with high probability, all the vertices in $V_h^d$ are in $\mathcal{V}^{d}_h$ and all the vertices in $V\setminus V_h^d$ are in $\mathcal{V}^{d}_l$.
We then take $V_l^d=V\setminus V_h^d$.
The overall complexity of this approach is 
$
\tilde{O}\left(\sqrt{n\times \frac{m}{n^d}}\times \sqrt{\frac{n}{n^{d}}}\right)=\tilde{O}(n^{1-d}\sqrt{m})
$
queries, since $|\mathcal{V}^d_h|=O(m/{n^d})$.
\end{proof}

In the remaining of the section we assume that the algorithm of Proposition~\ref{prop:classification} outputs a correct classification (i.e., $V^{d}_h\subseteq \mathcal{V}^{d}_h$ and $V^{d}_l\subseteq \mathcal{V}^{d}_l$), which happens with high probability. In particular we assume that $|V^d_h|=O(m/{n^d})$. We will say that a vertex $v\in V$ is $d$-high if $v\in V^{d}_h$, and say it is $d$-low if $v\in V^{d}_l$. Once the vertices have been classified, checking if $G$ has a triangle can be divided into four subproblems: checking if $G$ has a triangle with three $d$-low vertices, checking if $G$ has a triangle with two $d$-low vertices and one $d$-high vertex, checking if $G$ has a triangle with one $d$-low-degree vertex and two $d$-high vertices, and checking if $G$ has a triangle with three high-degree triangles. We now present six procedures to handle these cases (for some cases we present more than one procedure to allow us to choose which procedure to use according to the value of $m$).

\begin{proposition}\label{prop:B2}
Let $a_1$, $k_1$ and $b_1$ be any constants such that $0 < a_1, k_1 < 1$ and $0 < b_1 < a_1$.
There exists a quantum algorithm that finds a triangle of $G$ consisting of three $d$-low vertices, if such a triangle exists, with high probability using $Q_2 = \tilde{O}(n + n^{k_1/2}m^{1/2} + n^{a_1 + d/2 + k_1 - 1/2} + n^{1/2 + d/2 + k_1 - a_1/2} + n^{3/2 + k_1/2 - a_1} + n^{1 + b_1 + d/2 - a_1} + n^{3/2 - b_1/2} + n^{3/2 - k_1/2})$ queries.
\end{proposition}

The proof of Proposition~\ref{prop:B2}
will use the following key lemma.
\begin{lemma}
\label{lem3}
Let $k$ be any constant such that $0 < k < 1$.
Suppose that $X$ is a set of size $|X| = \ceil{3n^{k}\log{n}}$
obtained by taking uniformly at random vertices from~$V^{d}_l$.
Then, with probability at least
\[
1 - \frac{1}{n\exp(\frac{231}{10}n^{d + k - 1}\log{n})},
\]
the inequality
\[
|N_G(v) \cap X| < \frac{33}{10}n^{d + k - 1}\log{n} + 2\log{n}
\]
holds for all vertices $v \in V^{d}_l$.
\end{lemma}
\begin{proof}
Without loss of generality assume that $|N_G(v) \cap V^{d}_l| = \frac{11}{10}n^{d}$ for any $v \in V^{d}_l
$ (remember that we assume that $V^{d}_l\subseteq \mathcal{V}^{d}_l$). The quantity
$|N_G(v) \cap X|$ is a random variable distributed according to the hypergeometric distribution.
The expected value of $|N_G(v) \cap X|$ is $\mu = |N_G(v) \cap V^{d}_l| \times \frac{|X|}{|V^{d}_l|} = \frac{33}{10}\times n^{d + k - 1}\log{n}$.
By Corollary 2.4 and Theorem 2.10 of \cite{  janson2011random  }, $\Pr[|(N_G(v) \cap X| \geq 7\mu + 2\log{n}] \leq  \exp(-7\mu)\frac{1}{n^{2}}$.
Thus the inequality
\[
|N_G(v) \cap X| < \frac{33}{10}n^{d + k - 1}\log{n} + 2\log{n}
\]
holds for all $v \in V^{d}_l$ with probability at least
\[
1 - |V^{d}_l| \times \frac{1}{n^{2}\exp(\frac{231}{10}n^{d+k-1}\log{n})} = 1 - \frac{1}{n\exp(\frac{231}{10}n^{d+k-1}\log{n})},
\]as claimed.
\end{proof}

\begin{proof}[Proof of Proposition \ref{prop:B2}]
We adapt the algorithm for the dense case presented in Section \ref{sec:dense}.
We take $V_1 = V_2 = V^{d}_l$, and $X \subseteq V_1$ of size $|X|=\ceil{3|V_1|^{k_1}\log n}$.

We replace the first step of the algorithm, which checks if there exists a triangle of $G$ with a vertex in $X$ and two vertices in $V_2$,  by the following procedure based on \cite{Buhrman+SICOMP05}.
We take a random edge $\{u, v\} \in \E(V_2)\cap E$ and then try to find a vertex $w$ from $X$ such that $\{u,v,w\}$ is a triangle of $G$. Note that this can be implemented using two Grover searches in $\tilde O(\sqrt{|\E(V_2)|/|\E(V_2)\cap E|}+\sqrt{|X|})$ queries, and that in the worst case (i.e., when there is only one triangle) the success probability of this approach is $\Theta (1/|\E(V_2)\cap E|)$. Using amplitude amplification we can then check with high probability the existence of such a triangle with total query complexity
\begin{equation}\label{eq2}
\tilde O\left(\sqrt{|\E(V_2)\cap E|}\times (\sqrt{|\E(V_2)|/|\E(V_2)\cap E|}+\sqrt{|X|})\right) = \tilde{O}(n + \sqrt{n^{k_1}m}).
\end{equation}

We now show how to adapt the second step of the algorithm presented in Section \ref{sec:dense} to exploit the sparsity of the graph.
First, as observed in \cite{LeGallFOCS14}, the cost of estimating the size of $\Delta_G(X,A,w)$ can be reduced to $\tilde O(\sqrt{n^{k_1}})$ queries by using quantum counting instead of random sampling (quantum counting was not used in \cite{LeGallFOCS14} since it did not result in any speed-up for the dense case, but for the sparse case this is necessary).
We now describe our main ideas to exploit the sparsity of the graph, and show how to reduce the cost of two quantum walks.

First, we describe how to reduce the setup cost $\mathsf{S}$ and the update cost $\mathsf{U}$ as follows.
By Lemma \ref{lem3}, we know that $|N_G(v) \cap X| < t$ for all $v \in V_2$, where $t=\frac{33}{10}n^{d_1 + k_1 - 1}\log{n} + 2\log{n}$.
Therefore we can use quantum enumeration to find all vertices in $N_G(v) \cap X$
with
\[
\tilde O\left(\sqrt{\frac{|X|}{t}} + \cdots + \sqrt{\frac{|X|}{1}}\right) = \tilde O(\sqrt{|X|t})
\]
queries.
The setup cost $\mathsf{S}$ is thus
\[
\mathsf{S} \!= \tilde O\left(|A| \times \sqrt{|X|t}\right) = \tilde O(n^{a_1 + k_1 + d/2 - 1/2})
\]
queries, and
the update cost $\mathsf{U} = \tilde O(\sqrt{|X|t}) = \tilde O(n^{k_1 + d/2 - 1/2})$ queries.

Next, we describe how to reduce the setup cost $\mathsf{S}'$. This set up requires to obtain the couple $(v, e_v)$ for each $v \in B$, where where $w$ is a fixed vertex in $V_1$, $e_v=1$ if $\{v,w\}\in E$ and $e_v=0$ if $\{v,w\}\notin E$.
Let $\mu(w) = n^{d} \times \frac{|B|}{|V|} = n^{d + b_1 - 1}$ be the average of $|N_G(w) \cap B|$ over all $B$.
We use quantum enumeration to find at most $10 \times \mu(w)$ vertices in $N_G(w) \cap B$ from $B$.
Thus the cost of this procedure is
\[
\mathsf{S}' = \tilde O\left(\sqrt{\frac{|B|}{|\mu(w)|}} + \cdots \sqrt{\frac{|B|}{1}} \right) = \tilde O(\sqrt{|B|\times |\mu(w)}|) = \tilde O(n^{b_1 +  d/2 - 1/2})
\]
queries. Note that this procedure will not correctly prepare the database for all $B$'s (since $|N_G(w) \cap B|$ may exceeds $10 \times \mu(w)$ for some $B$'s); it will prepare correctly the database only for a large fraction of the $B$'s. This is nevertheless not a problem since the initial state of the quantum walk is a uniform superposition of all the $B$'s: this procedure will thus prepare a state close enough to the ideal state, which will modify only in a negligible way the final success probability of the whole walk.

We also modify the definition of a marked state for the second walk (we add one condition). Namely, the state corresponding to a set $B \in \Ss(A, \ceil{|V_2|^{b_1}})$ will be marked if~$B$ satisfies the following three conditions:
\begin{itemize}
\setlength{\leftskip}{0.3cm}
\item[(i)]
there exist two vertices $v_1, v_2 \in B$ such that $\{v_1,v_2\}\in E$ (i.e., such that $\{v_1,v_2,w\}$ is a triangle of~$G$);
\item[(ii)]
$|\mydelta{B,w}|\le 10\times |V_2|^{2(b_1-a_1)}\times \delta(X,A,w)$;
\item[(iii)]
$|N_G(w) \cap B| \leq 10 \times \mu(w)$.
\end{itemize}
It is easy to show that adding the third condition does not change significantly the fraction of marked states:
\begin{eqnarray*}
\varepsilon' &=& 
\Pr[v_1 \in B \ {\rm and} \  v_2 \in B \ {\rm and} \ |N_G(w) \cap B| \leq 10\mu(w) ] \\  &=&
\Pr[v_1 \!\in\! B] \Pr[v_2 \!\in\! B \mid v_1 \!\in\! B] \Pr[|N_G(w) \cap B| \leq 10\mu(w) \mid v_1 \!\in \! B \ {\rm and} \ v_2 \! \in \! B ] \\ &=&
\Pr[v_1 \in B] \Pr[v_2 \in B \mid v_1 \in B] \Pr[|N_G(w) \cap B'| \leq 10\mu(w) - 2 ] \\ &\geq& 
\left(1 - \Pr[|N_G(w) \cap B'| \geq 10\mu(w) - 2 ]\right) \times \Pr[v_1 \in B] \Pr[v_2 \in B \mid v_1 \in B]  \\ &\geq&
\left(1 - \Pr[|N_G(w) \cap B'| \geq 10\mu(w)' - 2 ]\right) \times \Pr[v_1 \in B] \Pr[v_2 \in B \mid v_1 \in B]  \\ &\geq&
\left(1 - \Pr[|N_G(w) \cap B'| \geq 2\mu(w)' ]\right) \times \Pr[v_1 \in B] \Pr[v_2 \in B \mid v_1 \in B]  \\ &\geq&
\left(1 - \frac{1}{2}\right) \times \Pr[v_1 \in B] \Pr[v_2 \in B \mid v_1 \in B]  = \Omega\left(n^{2(b_1 - a_1)}\right),
\end{eqnarray*}
where $B' \in \Ss(A\setminus \{v_1,v_2\}, \ceil{|V_2|^{b_1}} - 2)$ and $\mu(w)' = n^{d} \times \frac{|B|'}{|V|}$ be the average of $|N_G(w) \cap B'|$ over all $B'$.

The checking procedure of the second walk (and thus its cost $\mathsf{C}_w'$) is the same as in the dense case. 

By evaluating the performance of the walks as done for the dense case in Section \ref{sec:dense}, but replacing Expression (\ref{eq0}) by Expression (\ref{eq2}) and replacing in the evaluation of Expression (\ref{eq1}) the quantities $\mathsf{S}$, $\mathsf{U}$, $\mathsf{S'}$ and the cost of estimating $|\Delta_G(X,A,w)|$ by the expressions we just derived, we obtain the claimed query complexity for the whole algorithm. For instance,  the checking cost of the first walk
is
\begin{eqnarray*}
\mathsf{C} &=& \tilde{O}\!\left(\!\sqrt{|V_1|}\left(|V_1|^{k_1/2}+\mathsf{S}'\!+\!\frac{|V_2|^{b_1/2}\times \mathsf{U}'}{\sqrt{\varepsilon'}}\right)\!+\!\frac{|V_2|^{b_1-a_1}}{\sqrt{\varepsilon'}}\!\times\!\sqrt{\sum_{w\in V_1} |\Delta(X,A,w)|}\right)\!
\\ &=& \tilde{O}\left(n^{1/2 + k_1/2} + n^{1/2} \times \mathsf{S}' + n^{1/2 + a_1 - b_1/2} + \sqrt{\sum_{w\in V_1} |\Delta(X,A,w)|} \right)
\\ &=& \tilde{O}\left(n^{1/2 + k_1/2} + n^{b_1 + d/2} + n^{1/2 + a_1 - b_1/2} + n^{1/2 + a_1 - k_1/2}\right)
\end{eqnarray*}
queries.
\end{proof}

\begin{proposition}\label{prop:B3}
Let $a_2$, $k_2$ and $b_2$ be any constants such that $0 < a_2 < 1$, $1 < n^{k_2} < |V^{d}_h|$ and $0 < b_2 < a_2$.
A triangle of $G$ consisting of two $d$-low vertices and one $d$-high vertex can be detected with high probability using $Q_3 = \tilde{O}(n + n^{k_2/2}m^{1/2} + n^{a_2 + d + k_2}m^{-1/2} + n^{1 + d + k_2 - a_2/2}m^{-1/2} + n^{1 + k_2/2 - a_2 - d/2}m^{1/2} + n^{1 + b_2 -a_2 - d/2}m^{1/2} + n^{1 - b_2/2 - d/2}m^{1/2}$ $+  n^{1 - d/2- k_2/2}m^{1/2})$ queries.
\end{proposition}

\begin{proof}
We again adapt the algorithm for the dense case presented in Section~\ref{sec:dense}.
We take $V_1 = V^{d}_h$, $V_2 = V^{d}_l$, $a = a_2$, $b = b_2$ and $X \subseteq V_1$ of size $|X|=\ceil{3n^{k_2}\log n}$ (i.e., choose $k$ such that $|V_1|^k=n^{k_2}$).
The algorithm is exactly the same as the algorithm of Proposition \ref{prop:B2}, except that this time we cannot use the sparsity of the graph (since the vertices in $V_1$ are not $d$-low anymore) to reduce $\mathsf{S}'$. Instead, we use the same set up procedure as for the dense case in the second walk.
\end{proof}

\begin{proposition}\label{prop:B4}
Let $a_3$, $k_3$ and $b_3$ be constants such that $1 < n^{a_3} < |V^{d}_h|$, $0 < k_3 < 1$ and $0 < b_3 < a_3$.
A triangle of $G$ consisting of two $d$-high vertices and one $d$-low vertex can be detected with high probability using $Q_4 = \tilde{O}(n + n^{k_3/2}m^{1/2} + n^{a_3 + k_3} + n^{k_3 - a_3/2 - d}m + n^{1/2 + k_3/2 - a_3 - d}m + n^{b_3 - a_3 - d/2}m + n^{1/2 - b_3/2 - d}m + n^{1/2 - d - k_3/2}m )$ queries.
\end{proposition}

\begin{proof}
We adapt the algorithm for the dense case to the case we are considering. This time we choose $V_1 = V^{d}_l$, $V_2 = V^{d}_h$, $k=k_3$, take $a$ such that $|V_2|^a=n^{a_3}$ and~$b$ such that $|V_2|^b=n^{b_3}$. The algorithm is again almost the same as the algorithm of Proposition \ref{prop:B2}, except that this time we cannot use the sparsity of the graph to reduce $\mathsf{S}$ and $\mathsf{U}$ since the vertices in $V_2$ are not $d$-low anymore (but we can use the sparsity to reduce $\mathsf{S}'$ exactly as in the algorithm of Proposition~\ref{prop:B2}). Instead, we use the same set up and checking procedures as for the dense case in the first walk.
\end{proof}

\begin{proposition}\label{prop:B5}
A triangle of $G$ consisting of three $d$-high vertices can be detected with high probability using $Q_5 = \tilde O((m/n^d)^{5/4})$ queries.
\end{proposition}

\begin{proof}
We simply apply the original algorithm by Le Gall \cite{LeGallFOCS14} for dense triangle finding over the subgraph of $G$ induced by the vertices in $V^{d}_h$, and use the bound $|V_h^d|=O(m/n^d)$.
\end{proof}

\begin{proposition}\label{prop:B6}
Let $b_4$ be any constant such that $0 < b_4 < 1$. A triangle of~$G$ consisting of three $d$-low vertices can be detected with high probability using $Q_6 = \tilde{O}(n^{b_4 + d/2} + n^{3/2 - b_4/2} + n^{1/2 + d})$ queries.
\end{proposition}

\begin{proof}
We take $V_1 = V_2 = V^{d}_l$, and adapt the algorithm for the dense case as in Proposition~\ref{prop:B2}, but we choose $X=\emptyset$ and $a=1$, i.e., we do not perform the first step of the algorithm and do not perform the first walk in the second step. That is, we only perform the second walk of the algorithm, with parameter $b = b_4$.
The sparsity of the graph can again be used to reduce $\mathsf{S}'$, exactly as in the algorithm of Proposition~\ref{prop:B2}. The main difference is that we now use directly the sparsity of the graph for the checking step $\mathsf{C}'$: instead of performing a Grover search over $\mydelta{B,w}$, as in the dense case, we simply do a Grover search over $\E(N(w)\cap B)$, at cost
\[
\tilde O(\sqrt{|\E(N(w)\cap B)|})=\tilde O(10\times \mu(w)),
\]
which gives a new upper bound
$\mathsf{C}'_w$. Replacing in the analysis of Section \ref{sec:dense} the quantities $\mathsf{S}'$ and $\mathsf{C}'_w$ by these upper bounds, we obtain the claimed query complexity
\begin{eqnarray*}
C &=& \tilde{O}\left(\sqrt{|V_1|}\left(\mathsf{S}' + \sqrt{\frac{1}{\varepsilon'}}\left(|V_2|^{b_4/2} \times \mathsf{U}' + \mathsf{C}'_w\right)\right)\right) \\
&=& \tilde{O}(n^{1/2} \times \mathsf{S}' + n^{3/2 - b_4/2} + n^{3/2 - b_4} \times \mathsf{C}'_w ) \\
&=& \tilde O(n^{b_4 + d/2} \! +\! n^{3/2 - b_4/2} \!+\! n^{3/2 - b_4} \!\times\! \mu(w)) 
= \tilde O(n^{b_4 + d/2} + n^{3/2 - b_4/2} + n^{1/2 + d})
\end{eqnarray*}
\end{proof}

\begin{proposition}\label{prop:B7}
A triangle consisting of at least one $d$-high vertex
can be detected with high probability using $Q_7 = O(n + n^{-d/2}m)$ queries.
\end{proposition}

\begin{proof}
We use an algorithm similar to the procedure described in the first part of the proof of Proposition~\ref{prop:B2}, based on \cite{Buhrman+SICOMP05}.
We take a random edge $\{u, v\} \in E$ and then try to find a vertex $w$ from $V_h^d$ such that $\{u,v,w\}$ is a triangle of~$G$. This can be implemented using two Grover searches in $\tilde O(\sqrt{n^2/m}+\sqrt{|V_h^d|})$ queries, and that in the worst case (i.e., when there is only one triangle) the success probability of this approach is $\Theta (1/m)$. Using amplitude amplification we can then check with high probability the existence of such a triangle with total query complexity
\[
\tilde{O}\left(\sqrt{m}\times \left(\sqrt{n^2/m}+\sqrt{|V_h^d|}\right)\right) = \tilde{O}\left(n + \frac{m}{\sqrt{n^d}}\right),
\]
since $|V_h^d|=O(m/n^d)$.
\end{proof}

We are now ready to prove Theorem \ref{th:main}.
\begin{proof}[Proof of Theorem \ref{th:main}]
From Propositions 1, 2, 3, 4, 5, 6 and 7 and the discussion before Proposition 2, the query complexity of our whole algorithm is
\[
\min\big[(Q_1+ Q_6+Q_7) , \:(Q_1 + Q_3 + Q_4 + Q_5 + Q_6), \:(Q_1 + Q_2 + Q_3 + Q_4 + Q_5) \big].
\]
We write $m = n^{\ell}$, for $0\le \ell\le 2$, and optimize  below the parameters.

If $\frac{7}{6} \leq \ell \leq \frac{7}{5}$, the query complexity is upper bounded by 
\[
Q_1+ Q_6+Q_7=\tilde O(n^{1 + {\ell/2} -d} + n^{{3/2} - {b_4/2}} + n^{b_4 + d/2}), 
\]
which is optimized by taking $b_4 = 1 - \frac{\ell}{7}$ and $d = \frac{3\ell}{7}$, giving the upper bound $\tilde O(n^{1 + \ell/14})$.

If $\frac{7}{5} \leq \ell \leq \frac{3}{2}$, the query complexity is upper bounded by
\[
Q_1+ Q_6+Q_7=\tilde O(n^{3/2 - b_4/2} + n^{1/2 + d} + n^{\ell - d/2}),
\]
which is optimized by taking $b_4 = \frac{8}{3} - \frac{4\ell}{3}$ and $d = \frac{2\ell}{3} - \frac{1}{3}$, giving the upper bound $\tilde O(n^{1/6 + 2\ell/3})$.

If $\frac{3}{2} \leq \ell \leq \frac{13}{8}$, the query complexity is upper bounded by 
\begin{multline*}
Q_1 + Q_3 + Q_4 + Q_5 + Q_6 = \tilde O(n^{3/2 - b_4/2} + n^{1/2 + d} + n^{a_2 + d + k_2 - \ell/2} \\
+ n^{1 + b_2 + \ell/2 - a_2 - d/2} + n^{1 + \ell/2 - b_2/2 - d/2} + n^{1 + \ell/2 - d/2 - k_2/2} \\
+ n^{b_3 + \ell - a_3 - d/2} + n^{1/2 + \ell - b_3/2 - d} + n^{1/2 + \ell - d - k_3/2}),
\end{multline*}
which is optimized by taking $a_2 = \frac{3}{10} + \frac{3\ell}{10}$, $a_3 = \frac{23\ell}{15} - \frac{59}{30}$, $b_2 = k_2 = \frac{1}{5} + \frac{\ell}{5}$, $b_3 = k_3 = \frac{14\ell}{15} - \frac{16}{15}$, $b_4 = \frac{22}{15} - \frac{8\ell}{15}$ and $d = \frac{4}{15} + \frac{4\ell}{15}$, giving the upper bound $\tilde O(n^{23/30 + 4\ell/15})$.

If $\frac{13}{8} \leq \ell \leq 2$, the query complexity is upper bounded by
\begin{multline*}
Q_1 + Q_2 + Q_3 + Q_4 + Q_5 = \tilde O(n^{a_1 + {d/2} + k_1 - {1/2}} + n^{1  + b_1 +  {d/2} -  a_1} \\  
 +  n^{{3/2}  -  {b_1/2}} +  n^{{3/2}  -   {k_1/2}}  +  n^{a_2  +  d  +  k_2  -   {\ell/2}}  +  n^{1 + b_2 + {\ell/2} - a_2 - {d/2}} + n^{1 + {\ell/2} - {b_2/2} - {d/2}} \\
+ n^{1 + {\ell/2} - {d/2} - {k_2/2}} + n^{b_3 + \ell - a_3 - {d/2}} + n^{{1/2} + \ell - {b_3/2} - d} + n^{{1/2} + \ell - d - {k_3/2}} ),
\end{multline*}
which is optimized by taking $a_1 =  \frac{3}{4}$, $a_2 = \frac{19}{20} - \frac{\ell}{10}$, $a_3 = \frac{3\ell}{5} - \frac{9}{20}$, $b_1  =  k_1 = \frac{31}{30} -  \frac{4\ell}{15}$, $b_2 =  k_2 = \frac{19}{30} -  \frac{\ell}{15}$, $b_3 = k_3 = \frac{7}{30} + \frac{2\ell}{15}$ and $d  =  \frac{4\ell}{5} - \frac{3}{5}$, giving the upper bound $\tilde O(n^{59/60 + 2\ell/15})$.

For the case $\ell\le \frac{7}{6}$ we can obtain a better upper bound by using the $O(n+\sqrt{nm})$-query algorithm by Buhrman et al. \cite{Buhrman+SICOMP05}. This upper bound actually corresponds to a degenerate case appearing in our approach: the case $d=0$. Indeed, observe that without loss of generality we can assume that $\deg(v)\ge 1$ for all vertices $v\in V$ (for instance by adding dummy vertices to the graph). In this case we have $\mathcal{V}^{d}_h=V$ for $d=0$, which means that we do not need to apply the algorithm of Proposition \ref{prop:classification} in order to obtain a classification: we simply output $V^{0}_h=V$ and $V^0_l=\emptyset$ (i.e., all the vertices of the graph are $0$-high). The only type of triangles we need to consider is triangles with three 0-high vertices, which can be found with complexity $O(n+\sqrt{nm})$ by Proposition \ref{prop:B7} (in this case the algorithm of Proposition \ref{prop:B7} is exactly the same as the algorithm in \cite{Buhrman+SICOMP05}).
\end{proof}

\section*{Acknowkedgments}
The authors are grateful to Mathieu Lauri{\`e}re, Keiji Matsumoto, Harumichi Nishimura and Seiichiro Tani for helpful comments. This work is supported by the Grant-in-Aid for Young Scientists~(B)~No.~24700005, the Grant-in-Aid for Scientific Research~(A)~No.~24240001 of the Japan Society for the Promotion of Science and the Grant-in-Aid for Scientific Research on Innovative Areas~No.~24106009 of the Ministry of Education, Culture, Sports, Science and Technology in Japan.

\end{document}